\documentclass[12pt,a4paper]{article}
\usepackage[top=2.95cm,bottom=2.95cm,left=2.2cm,right=2.2cm]{geometry}
\usepackage{latexsym,amsmath,amsfonts,amssymb,amsthm,amscd,mathrsfs}

\newcommand{\be}{\begin{equation}}
\newcommand{\ee}{\end{equation}}
\newcommand{\bea}{\begin{eqnarray}}
\newcommand{\eea}{\end{eqnarray}}

\numberwithin{equation}{section}
\newcounter{thmcounter}
\numberwithin{thmcounter}{section}
\theoremstyle{definition}
\newtheorem*{acknowledgements}{Acknowledgements}
\newtheorem{remark}[thmcounter]{Remark}
\theoremstyle{plain}
\newtheorem{lemma}[thmcounter]{Lemma}
\newtheorem{proposition}[thmcounter]{Proposition}
\newtheorem{theorem}[thmcounter]{Theorem}

\def\1{{\boldsymbol 1}}                     %
\def\0{{\boldsymbol 0}}                     %
\def\bC{{\boldsymbol C}}                    %
\def\bJ{{\boldsymbol J}}                    %
\def\C{\mathbb{C}}                          %
\def\R{\mathbb{R}}                          %
\def\T{\mathbb{T}}                          %
\def\Z{\mathbb{Z}}                          %
\def\cC{{\mathcal C}}                       %
\def\cH{{\mathcal H}}                       %
\def\cL{{\mathcal L}}                       %
\def\cZ{{\mathcal Z}}                       %
\def\su{\mathfrak{su}}                      %
\def\tr{\mathrm{tr}}                        %
\def\diag{\mathrm{diag}}                    %
\def\ri{\mathrm{i}}                         %
\def\BC{\mathrm{BC}}                        %
\def\SB{\mathrm{SB}}                        %
\def\SL{\mathrm{SL}}                        %
\def\SO{\mathrm{SO}}                        %
\def\UN{\mathrm{U}}                         %
\def\SU{\mathrm{SU}}                        %

\begin{document}

\begin{center}
{\large\bf
The full phase space of a model in the Calogero-Ruijsenaars family}
\end{center}

\vspace{0.2cm}

\medskip
\begin{center}
L.~Feh\'er${}^{a,b}$ and T.F.~G\"orbe${}^a$\\

\bigskip
${}^a$Department of Theoretical Physics, University of Szeged\\
Tisza Lajos krt 84-86, H-6720 Szeged, Hungary\\
e-mail: tfgorbe@physx.u-szeged.hu

\medskip
${}^b$Department of Theoretical Physics, WIGNER RCP, RMKI\\
H-1525 Budapest, P.O.B.~49, Hungary\\
e-mail: lfeher@physx.u-szeged.hu
\end{center}

\vspace{0.2cm}

\medskip
\begin{abstract}
We complete the recent derivation of a Ruijsenaars type system that arises as
a reduction of the natural free system on the Heisenberg double of $\SU(n,n)$.
The previous analysis by Marshall focused on a dense open submanifold of the
reduced phase space, and here we describe the full phase space wherein
Liouville integrability of the system holds by construction.
\end{abstract}

\newpage

\section{Introduction}
\label{sec:1}

The method of Hamiltonian reduction belongs to the set of standard toolkits
applicable to study a great variety of problems ranging from geometric mechanics
to field theory and harmonic analysis \cite{MR,E}. It is especially useful in the theory of
integrable Hamiltonian systems \cite{BBT}, where one of the maxims is that one should view
the systems of interests as reductions of obviously solvable `free' systems \cite{OP1}.
This is often advantageous, for example since the reduction produces global phase
spaces on which the reduced free flows are automatically complete, which is an
indispensable property of any integrable system.

An interesting application of this method appeared in the recent paper \cite{M}
by Marshall, where a deformation of the classical hyperbolic $\BC_n$ Sutherland
system \cite{OP1} was derived by reduction of a free system on the Heisenberg double \cite{STS}
of the Poisson-Lie group $\SU(n,n)$. In a closely related work \cite{FG}, we investigated
the analogous reduction of the Heisenberg double of $\SU(2n)$ and thereby obtained
a deformation of the trigonometric variant of the $\BC_n$ Sutherland system.
In the course of our analysis \cite{FG} we noticed that Marshall's considerations
were restricted to a proper open submanifold of the reduced phase space.
Although this submanifold forms a dense subset, the reduced Hamiltonian flows are
complete only on the full phase space. The goal of this paper is to provide a
globally valid model of the relevant reduced phase space.

Our motivation basically stems from the fact that the global description of the
reduced phase space is a necessary ingredient of the characterization of any
Hamiltonian reduction. Besides, we also intend this complement to the papers
\cite{M,FG} to serve as a starting point for a future work where the duality
aspects of the $\SU(n,n)$ and the $\SU(2n)$ cases should be treated together.
Based on experience \cite{FK1,P1}, ideas for finding the pertinent dual systems are
readily available, but their technical implementation poses a challenging
open problem.

Below, we concentrate on the essential points referring to \cite{M,FG} for more
details. In Section \ref{sec:2} the necessary preliminaries are summarized.
Then in Section \ref{sec:3} we describe the reduced phase space. In Subsection
\ref{subsec:3.1}, we present the local picture of Marshall using a shortcut that
leads to it, and give the global picture in Subsection \ref{subsec:3.2}.
The reader may go directly  to Theorem \ref{thm:3.6} to see the result.
The reduced Hamiltonians and their integrability is briefly discussed in the last
section, and two appendices are included to help readability.
Appendix \ref{sec:A} contains some auxiliary explicit formulas,
while Appendix \ref{sec:B} details a property of the reduced Hamiltonians.

\section{Definitions and first steps}
\label{sec:2}

All of the material presented in this section is adapted from \cite{M} and \cite{FG}.

Fix an integer $n>1$ and introduce the group
\be
\SU(n,n)=\{g\in\SL(2n,\C)\mid g^\dag\bJ g=\bJ\}
\label{T1}
\ee
with $\bJ=\diag(\1_n,-\1_n)$. Then consider the open submanifold
$\SL(2n,\C)'\subset\SL(2n,\C)$ consisting of those elements, $K$,
that admit both Iwasawa-like decompositions of the form
\be
K=g_Lb_R^{-1}=b_Lg_R^{-1},\qquad g_L,g_R\in\SU(n,n),\quad b_L,b_R\in\SB(2n),
\label{T2}
\ee
where $\SB(2n)<\SL(2n,\C)$ is the group of upper triangular matrices having
positive entries along the diagonal. Both decompositions are unique and the
constituent factors depend smoothly on $K\in\SL(2n,\C)'$.
The manifold $\SL(2n,\C)'$ carries the symplectic form \cite{AM}
\be
\omega=\frac{1}{2}\Im\tr(db_Lb_L^{-1}\wedge dg_Lg_L^{-1})
+\frac{1}{2}\Im\tr(db_Rb_R^{-1}\wedge dg_Rg_R^{-1}).
\label{T3}
\ee
On this symplectic manifold, which is a symplectic leaf of the Heisenberg double
of the Poisson-Lie group $\SU(n,n)$, one has the pairwise Poisson
commuting Hamiltonians
\be
\cH_j(K)=\frac{1}{2j}\tr(K\bJ K^\dag\bJ)^j,\qquad j\in\Z^\ast.
\label{T4}
\ee
They generate complete flows that can be written down explicitly (see
Section \ref{sec:4}). We are concerned with a reduction of these Hamiltonians
based on the symmetry group $G_+\times G_+$, where
\be
G_+=\mathrm{S}(\UN(n)\times\UN(n))<\SU(n,n)
\label{T5}
\ee
is the block-diagonal subgroup. Throughout, we refer to the obvious $2\times 2$
block-matrix structure corresponding to $\bJ$. The action of $G_+\times G_+$
on $\SL(2n,\C)'$ is given by the map
\be
G_+\times G_+\times\SL(2n,\C)'\to\SL(2n,\C)'
\label{T6}
\ee
that works according to
\be
(\eta_L,\eta_R,K)\mapsto\eta_LK\eta_R^{-1}.
\label{T7}
\ee
One can check that this map is well-defined, i.e.~$\eta_LK\eta_R^{-1}$ stays in
$\SL(2n,\C)'$, and has the Poisson property with respect to the product Poisson
structure on the left-hand side \cite{STS,AM}, where on $G_+$ the standard Sklyanin bracket
is used and the Poisson structure on $\SL(2n,\C)'$ is engendered by $\omega$.
Moreover, this $G_+\times G_+$ action is associated with a momentum map in the
sense of Lu \cite{Lu}. The momentum map can be written as
$\Phi_+\colon\SL(2n,\C)'\to G_+^\ast\times G_+^\ast$, where $G_+^\ast$ is the
subgroup of $\SB(2n)$ containing the elements with vanishing off-diagonal blocks.
By utilizing the obvious projection $\pi\colon\SB(2n)\to G_+^\ast$, which replaces
the block off-diagonal components by zeroes, the momentum map obeys the formula
\be
\Phi_+(K)=\big(\pi(b_L),\pi(b_R)\big).
\label{T8}
\ee
The Hamiltonians $\cH_j$ \eqref{T4} are invariant with respect to the symmetry
group $G_+\times G_+$ and $\Phi_+$ is constant along their flows.

The general theory \cite{Lu} ensures that one can now perform Marsden-Weinstein
type reduction. This amounts to imposing the constraint
\be
\Phi_+(K)=\mu=(\mu_L,\mu_R)
\label{T9}
\ee
with some constant $\mu\in G_+^\ast\times G_+^\ast$ and then taking the quotient
of $\Phi_+^{-1}(\mu)$ by the corresponding isotropy group, denoted below as $G_\mu$.

We pick the following value $\mu$ of the momentum map,
\be
\mu_L=\begin{bmatrix}e^u\nu(x)&\0_n\\\0_n&e^{-u}\1_n\end{bmatrix},\quad
\mu_R=\begin{bmatrix}e^v\1_n&\0_n\\\0_n&e^{-v}\1_n\end{bmatrix},
\label{T10}
\ee
where $u$, $v$, and $x$ are real constants satisfying
\be
u +v\neq 0,\qquad x> 0,
\label{T11}
\ee
and $\nu(x)$ is the $n\times n$ upper triangular matrix defined by
\be
\nu(x)_{jj}=1,\quad\nu(x)_{jk}=(1-e^{-x})e^{\frac{(k-j)x}{2}},\quad j<k.
\label{T12}
\ee
The essential property of $\nu(x)$ is that $\nu(x)\nu(x)^\dag$ has only two
different eigenvalues, one of them with multiplicity $1$.
(This also holds for $x<0$ and our assumption $x>0$
only serves to keep the text shorter.)
The corresponding isotropy group $G_\mu$ is
\be
G_\mu=G_+(\mu_L)\times G_+,
\label{T13}
\ee
where the elements $\eta_L\in G_+(\mu_L)$ have the form
\be
\eta_L=\begin{bmatrix}\eta_L(1)&\0_n\\\0_n&\eta_L(2)\end{bmatrix}
\label{T14}
\ee
with $\eta_L(1)\in\UN(n)$ satisfying
\be
\eta_L(1)\nu(x)\nu(x)^\dag\eta_L(1)^{-1}=\nu(x)\nu(x)^\dag
\label{T15}
\ee
 and $\eta_L(2)\in\UN(n)$, coupled by $\det(\eta_L)=1$.
It will turn out that the reduced phase space
\be
M=\Phi_+^{-1}(\mu)/G_\mu
\label{T16}
\ee
is a smooth manifold. Our task is to characterize this manifold, which carries the reduced
symplectic form $\omega_M$ defined by the relation
 \be
\iota_\mu^\ast\omega =\pi_\mu^\ast\omega_M,
\label{defred}\ee
where $\iota_\mu\colon\Phi_+^{-1}(\mu)\to\SL(2n,\C)'$ is the tautological
injection and $\pi_\mu\colon\Phi_+^{-1}(\mu)\to M$ is the natural projection.

Consider the following central subgroup $\Z_{2n}$ of $G_+ \times G_+$,
\be
\Z_{2n} = \{ (w \1_{2n}, w\1_{2n})\mid w\in \C,\,\, w^{2n}=1\},
\label{Z2n}\ee
which acts trivially according to \eqref{T7} and is contained in $G_\mu$.
Later we shall refer to the factor group
\be
\bar G_\mu = G_\mu/\Z_{2n}
\label{barGmu}\ee
as the `effective gauge group'. Obviously, we have
$\Phi_+^{-1}(\mu)/G_\mu = \Phi_+^{-1}(\mu)/\bar G_\mu$.

In the end, we shall obtain a model of the quotient space $M$ by explicitly
exhibiting a global cross-section of the orbits of $G_\mu$ in $\Phi_+^{-1}(\mu)$.
The construction uses the generalized Cartan decomposition of $\SU(n,n)$, which
says that every $g\in \SU(n,n)$ can be written as
\be
g=g_+\begin{bmatrix}\cosh q&\sinh q\\ \sinh q&\cosh q\end{bmatrix}h_+,
\label{T17}
\ee
where $g_+,h_+\in G_+$ and $q=\diag(q_1,\dots,q_n)$ is a real diagonal matrix
verifying
\be
q_1\geq\dots\geq q_n\geq 0.
\label{T18}
\ee
The components $q_i$ are uniquely determined by $g$, and yield smooth functions
on the locus where they are all distinct. In what follows we shall often identify
diagonal matrices like $q$ with the corresponding elements of $\R^n$.

As the first step towards describing $M$, we apply the decomposition \eqref{T17}
to $g_L$ in $K=g_Lb_R^{-1}$ and impose the right-handed momentum constraint
$\pi(b_R)=\mu_R$. It is then easily seen that up to $G_\mu$-transformations every
element of $\Phi_+^{-1}(\mu)$ can be represented in the following form:
\be
K=\begin{bmatrix}\rho&\0_n\\\0_n&\1_n\end{bmatrix}
\begin{bmatrix}\cosh q& \sinh q\\ \sinh q&\cosh q\end{bmatrix}
\begin{bmatrix}e^{-v}\1_n&\alpha\\\0_n&e^v\1_n\end{bmatrix}.
\label{T19}
\ee
Here $\rho\in\SU(n)$ and $\alpha$ is an $n\times n$ complex matrix. Referring to
the $2\times 2$ block-matrix notation, we introduce $\Omega=K_{22}$ and record
from \eqref{T19} that
\be
\Omega=(\sinh q)\alpha+e^v\cosh q.
\label{T20}
\ee
It will prove advantageous to seek for $\Omega$ in the polar-decomposed form,
\be
\Omega=\Lambda T,
\label{T21}
\ee
where $T\in\UN(n)$ and $\Lambda$ is a Hermitian, positive semi-definite matrix.

The next step is to implement the left-handed momentum constraint $\pi(b_L)=\mu_L$
by writing $K=b_Lg_R^{-1}$ with
\be
b_L=\begin{bmatrix}e^u\nu(x)&\chi\\\0_n&e^{-u}\1_n\end{bmatrix},
\label{T22}
\ee
where $\chi$ is an unknown $n\times n$ matrix. Then we inspect the components of
the $2\times 2$ block-matrix identity
\be
K\bJ K^\dag=b_L\bJ b_L^\dag,
\label{T23}
\ee
which results by substituting $K$ from \eqref{T19}. We find that the (22) component
of this identity is equivalent to
\be
\Omega\Omega^\dag=\Lambda^2=e^{-2u}\1_n+e^{-2v}(\sinh q)^2.
\label{T24}
\ee
This uniquely determines $\Lambda$ in terms of $q$ and also shows that $\Lambda$ is invertible.
An important consequence of the first condition in \eqref{T11} is that we must have
\be
q_n>0,
\label{T25}
\ee
and therefore $\sinh q$ is an invertible diagonal matrix.
Indeed, if $q_n=0$, then from \eqref{T20} and \eqref{T24} we would get
$(\Omega\Omega^\dag)_{nn}=e^{2v}=e^{-2u}$, which is excluded by \eqref{T11}.

By using the above relations, it is simple algebra to convert the (12)
and the (21) components of the identity \eqref{T23} into the equation
\be
\chi=\rho(\sinh q)^{-1}[e^{-u}\cosh q-e^{u+v}\Omega^\dag].
\label{T26}
\ee
Finally, the (11) entry of the identity \eqref{T23} translates into the following
crucial equation:
\be
\rho(\sinh q)^{-1}T^\dag(\sinh q)^2T(\sinh q)^{-1}\rho^\dag=\nu(x)\nu(x)^\dag.
\label{T27}
\ee
This is to be satisfied by $q$ subject to \eqref{T18}, \eqref{T25} and $T\in\UN(n)$,
$\rho\in\SU(n)$. After finding $q$, $T$ and $\rho$, one can reconstruct $K$
\eqref{T19} by applying the formulas derived above.

From our viewpoint, a key observation is that \eqref{T27} coincides completely with
equation (5.7) in the paper \cite{FK1}, where its general solution was found.
The correspondence between the notations used here and in \cite{FK1} is
\be
(\rho,T,\sinh q)\Longleftrightarrow(k_L,k_R^\dag,e^{\hat p}).
\label{T28}
\ee
For this reason, we introduce the new variable $\hat p\in\R^n$ by the definition
\be
\sinh q_k=e^{\hat p_k},\quad k=1,\dots,n.
\label{T29}
\ee
Because of \eqref{T18} and \eqref{T25}, the variables $\hat p_k$ satisfy
\be
\hat p_1\geq\dots\geq\hat p_n.
\label{T30}
\ee

We do not see an a priori reason why the very different reduction procedures led
to the same equation \eqref{T27} here and in \cite{FK1}. However, we are going to
take full advantage of this situation. We note that essentially every formula written in
this section appears in \cite{M} as well (with slightly different notations),
but in Marshall's work the previously obtained results about the solutions of
\eqref{T27} were not used.

\section{The reduced phase space}
\label{sec:3}

The statement of Proposition \ref{prop:3.3} characterizes a submanifold of $M$
\eqref{T16}, which was erroneously claimed in \cite{M} to be equal to $M$. After
describing this `local picture', we shall present a globally valid model of $M$.

\subsection{The local picture}
\label{subsec:3.1}

By applying results of \cite{FK1,FK2} in
the same way as in \cite{FG}, one can prove the following lemma.

\begin{lemma}
\label{lem:3.1}
The constraint surface $\Phi_+^{-1}(\mu)$ contains an element of the form
\eqref{T19} if and only if $\hat p$ defined by \eqref{T29} lies in the closed
polyhedron
\be
\bar\cC_x=\{\hat p\in\R^n\mid
\hat p_k-\hat p_{k+1}\geq x/2\ (k=1,\dots,n-1)\}.
\label{H1}
\ee
\end{lemma}

The polyhedron $\bar\cC_x$ is the closure of its interior, $\cC_x$, defined
by strict inequalities. We note that in \cite{M} the elements of the boundary
$\bar\cC_x\setminus\cC_x$ were omitted.

For any fixed $\hat p\in\bar\cC_x$, one can write down the solutions of \eqref{T27}
for $T$ and $\rho$ explicitly \cite{FK1}. By inserting those into the formula
\eqref{T19}, using the relations \eqref{T20}, \eqref{T21}, \eqref{T24} to determine
the matrix $\alpha$, one arrives at the next lemma. It refers to the $n\times n$
real matrices
\be
\theta(x,\hat p),\quad\zeta(x,\hat p),\quad\kappa(x),
\label{H2}
\ee
which belong to the group $\SO(n)$ and are defined by explicit formulas that can
be found, for example, at the beginning of Section 3.2 in \cite{FG}.
For the reader's convenience, we append these formulas at the end of the text.

\begin{proposition}
\label{prop:3.2}
For any parameters $u,v,x$ subject to \eqref{T11}, and variables
$\hat p\in\bar\cC_x$ and $e^{\ri\hat q}$ from the $n$-torus $\T_n$,
define the matrix
\be
K(\hat p,e^{\ri \hat q})
=\begin{bmatrix}\rho &\0_n\\\0_n&\1_n\end{bmatrix}
\begin{bmatrix}\sqrt{\1_n+e^{2\hat p}} & e^{\hat p}\\
e^{\hat p}&\sqrt{\1_n+e^{2\hat p}}\end{bmatrix}
\begin{bmatrix}e^{-v}\1_n&\alpha\\\0_n&e^v\1_n\end{bmatrix}
\label{H3}
\ee
by employing
\be
\rho=\rho(x,\hat p)=\kappa(x)\zeta(x,\hat p)^{-1}
\label{H4}
\ee
and
\be
\alpha=\alpha(x,u,v,\hat p, e^{\ri \hat q})
=e^{\ri\hat q}\sqrt{e^{-2u}e^{-2\hat p}+ e^{-2v}\1_n}\,\theta(x,\hat p)^{-1}
-e^v\sqrt{e^{-2\hat p}+\1_n}.
\label{H5}
\ee
Then $K(\hat p,e^{\ri\hat q})$ resides in the constraint surface $\Phi_+^{-1}(\mu)$
and the set
\be
S=\{K(\hat p,e^{\ri\hat q})\mid(\hat p,e^{\ri\hat q})\in\bar\cC_x\times\T_n\}
\label{H6}
\ee
intersects every orbit of $G_\mu$ in $\Phi_+^{-1}(\mu)$.
\end{proposition}

By arguing verbatim along the lines of \cite{FG}, and referring to \cite{M} for the
calculation of the reduced symplectic form, one can establish the validity of the
subsequent proposition.

\begin{proposition}
\label{prop:3.3}
The effective gauge group $\bar G_\mu$ \eqref{barGmu} acts freely
on $\Phi_+^{-1}(\mu)$ and thus the quotient space $M$
\eqref{T16} is a smooth manifold. The restriction of the natural projection
$\pi_\mu\colon\Phi_+^{-1}(\mu)\to M$ to
\be
S^o=\{K(\hat p,e^{\ri\hat q})\mid(\hat p,e^{\ri\hat q})\in\cC_x\times\T_n\}
\label{H7}
\ee
gives rise to a diffeomorphism between $\cC_x\times\T_n$ and the open, dense
submanifold of $M$ provided by $\pi_\mu(S^o)$. Taking $S^o$ as model of
$\pi_\mu(S^o)$, the corresponding restriction of the reduced symplectic
form $\omega_M$ becomes the Darboux form
\be
\omega_{S^o}=\sum_{k=1}^nd\hat q_k\wedge d\hat p_k.
\label{H8}
\ee
\end{proposition}

\begin{remark}
\label{rem:3.4}
In the formula \eqref{H3} $K(\hat p,e^{\ri\hat q})$ appears in the decomposed form
$K=g_Lb_R^{-1}$ and it is not immediately obvious that it belongs to $\SL(2n,\C)'$,
i.e., that it can be decomposed alternatively as $b_Lg_R^{-1}$.
However, by defining $b_L(\hat p, e^{\ri \hat q})\in \SB(2n)$ by the formula
\eqref{T22} using $\chi$ in \eqref{T26} with the change of variables
$\sinh q=e^{\hat p}$, the matrix $\rho$ as given above, and
$T=e^{\ri\hat q}\theta(x,\hat p)^{-1}$ that enters \eqref{H3}, we can verify that for
these elements $g_R^{-1}=b_L^{-1}K$ satisfies the defining relation of $\SU(n,n)$
\eqref{T1}, as required. The reader may perform this verification, which relies
only on the constraint equations displayed in Section \ref{sec:2}.
\end{remark}

\subsection{The global picture}
\label{subsec:3.2}

The train of thought leading to the construction below can be outlined as follows.
Proposition \ref{prop:3.3} tells us, in particular, that any $G_\mu$-orbit passing
through $S^o$ intersects $S^o$ in a single point. Direct inspection shows that the
analogous statement is false for $S\setminus S^o$, which corresponds to
$(\bar\cC_x\setminus\cC_x)\times\T_n$ in a one-to-one manner.
Thus a global model of $M$ should result by identifying those points of
$S\setminus S^o$ that lie on the same $G_\mu$-orbit. By using the bijective map
from $\bar\cC_x\times\T_n$ onto $S$ given by the formula \eqref{H3}, the desired
identification will be achieved by constructing such complex variables out of
$(\hat p,e^{\ri\hat q})\in\bar\cC_x\times\T_n$ that coincide precisely for gauge
equivalent elements of $S$.

Turning to the implementation of the above plan, we introduce the space of complex
variables
\be
\hat M_c=\C^{n-1}\times\C^\times,\qquad(\C^\times=\C\setminus\{0\}),
\label{hatM}
\ee
carrying the symplectic form
\be
\hat\omega_c=\ri\sum_{j=1}^{n-1}dz_j\wedge d\bar z_j
+\frac{\ri dz_n\wedge d\bar z_n}{2z_n\bar z_n}.
\label{hatom}
\ee
We also define the surjective map
\be
\hat\cZ_x\colon\bar\cC_x\times\T_n\to\hat M_c,\quad
(\hat p,e^{\ri\hat q})\mapsto z(\hat p,e^{\ri\hat q})
\label{H11}
\ee
by setting
\be
\begin{split}
z_j(\hat p,e^{\ri\hat q})
&=(\hat p_j-\hat p_{j+1}-x/2)^{\tfrac{1}{2}}\prod_{k=j+1}^ne^{\ri\hat q_k},
\quad j=1,\dots,n-1,\\
z_n(\hat p,e^{\ri\hat q})&=e^{-\hat p_1}\prod_{k=1}^ne^{\ri\hat q_k}.
\end{split}
\label{H12}
\ee
The restriction $\cZ_x$ of $\hat\cZ_x$ to $\cC_x\times\T_n$ is a diffeomorphism
onto the open subset
\be
\hat M_c^o=\bigg\{z\in\hat M_c\,\bigg\vert\,\prod_{j=1}^{n-1} z_j\neq 0\bigg\},
\label{H13}
\ee
and it verifies the relation
\be
\cZ_x^\ast\hat\omega_c=\sum_{k=1}^nd\hat q_k\wedge d\hat p_k.
\label{densehatom}
\ee
Thus we manufactured a change of variables
$\cC_x\times\T_n\longleftrightarrow\hat M_c^o$.
The inverse $\cZ_x^{-1}\colon\hat M_c^o\to\cC_x\times\T_n$ involves the functions
\be
\hat p_1(z)=-\log|z_n|,\quad
\hat p_j(z)=-\log|z_n|-\sum_{k=1}^{j-1}(|z_k|^2+x/2)\quad (j=2,\dots,n).
\ee
These extend smoothly to $\hat M_c$ wherein $\hat M_c^o$ sits as a dense
submanifold.

Now we state a lemma, which is a simple adaptation from \cite{FK1,RIMS95}.

\begin{lemma}
\label{lem:3.5}
By using the shorthand $\sigma_j=\prod_{k=j+1}^ne^{\ri\hat q_k}$ for $j=1,\ldots, n-1$ (cf.~\eqref{H12}), let us define
\be
\sigma_+( e^{\ri \hat q})=\diag(\sigma_1,\dots,\sigma_{n-1},1)
\quad\text{and}\quad
\sigma_-(e^{\ri \hat q})=\diag(1,\sigma_1^{-1},\dots,\sigma_{n-1}^{-1}).
\label{H16}
\ee
Then there exist unique smooth functions $\hat\zeta(x,z)$, $\hat\theta(x,z)$ and
$\hat \alpha(x,u,v, z)$ of $z\in \hat M_c$ that satisfy the following identities
for any $(\hat p,e^{\ri\hat q})\in\bar\cC_x\times\T_n$:
\begin{align}
\hat\zeta(x,z(\hat p,e^{\ri \hat q}))
&=\sigma_+(e^{\ri \hat q})\zeta(x,\hat p)\sigma_+( e^{\ri \hat q})^{-1},
\label{H19}\\
\hat\theta(x,z(\hat p,e^{\ri \hat q}))
&=\sigma_+(e^{\ri \hat q})\theta(x,\hat p) \sigma_-(e^{\ri \hat q}),
\label{H20}\\
\hat\alpha(x,u,v, z(\hat p,e^{\ri \hat q}))
&=\sigma_+(e^{\ri \hat q})
\alpha(x,u,v,\hat p, e^{\ri \hat q}) \sigma_+(e^{\ri \hat q})^{-1}.
\label{H21}
\end{align}
Here we refer to the functions on $\bar\cC_x\times\T_n$
displayed in equations \eqref{H2} and \eqref{H5}.
\end{lemma}

The explicit formulas of the functions on $\hat M_c$ that appear in the above
identities are easily found by first determining them on $\hat M_c^o$ using the
change of variables $\cZ_x$, and then noticing that they automatically extend to
$\hat M_c$. The expressions of the functions $\hat\zeta$ and $\hat\theta$, which
depend only on $z_1,\dots,z_{n-1}$, are the same as given in Definition 3.3
in \cite{FK1}. (For most purposes the above definitions and the formulas of
Appendix \ref{sec:A} suffice.) As for $\hat\alpha$, by defining
\be
\Delta(z)=\diag(z_n,e^{-\hat p_2(z)},\dots,e^{-\hat p_n(z)})
\label{H22}
\ee
we have
\be
\hat\alpha(x,u,v,z)=\sqrt{e^{-2v}e^{2\hat p(z)}+e^{-2u}\1_n}\,
\Delta(z)\hat \theta(x,z)^{-1}-e^v\sqrt{e^{-2\hat p(z)}+\1_n}
\label{H23}
\ee
that satisfies relation \eqref{H21} due to the identity
\be
\Delta(z(\hat p, e^{\ri \hat q}))=e^{-\hat p}e^{\ri \hat q}\sigma_+(e^{\ri \hat q})\sigma_-(e^{\ri \hat q}),
\qquad\forall (\hat p,e^{\ri\hat q})\in\bar\cC_x\times\T_n.
\label{H24}
\ee

With these preparations at hand, we can formulate the main result of this paper.

\begin{theorem}
\label{thm:3.6}
Define the smooth map $\hat K\colon\hat M_c\to\SL(2n,\C)'$ by the formula
\be
\hat K(z)
=\begin{bmatrix}\kappa(x)\hat\zeta(x,z)^{-1}&\0_n\\\0_n&\1_n\end{bmatrix}
\begin{bmatrix}\sqrt{\1_n+e^{2\hat p(z)}}&e^{\hat p(z)}\\
e^{\hat p(z)}&\sqrt{\1_n+e^{2\hat p(z)}}\end{bmatrix}
\begin{bmatrix}e^{-v}\1_n&\hat\alpha(x,u,v,z)\\\0_n&e^v\1_n\end{bmatrix}.
\label{H25}
\ee
The image of $\hat K$ belongs to the submanifold $\Phi_+^{-1}(\mu)$ and
the induced mapping $\pi_\mu \circ \hat K$, obtained by using the natural
projection $\pi_\mu\colon\Phi_+^{-1}(\mu)\to M=\Phi^{-1}(\mu)/G_\mu$,
is a symplectomorphism between $(\hat M_c,\hat\omega_c)$, defined by \eqref{hatM}-\eqref{hatom},
and the reduced phase space $(M,\omega_M)$.
\end{theorem}

\begin{proof}
We start by pointing out that for any
$(\hat p,e^{\ri\hat q})\in\bar\cC_x\times\T_n$ the identity
\be
\hat K(z(\hat p,e^{\ri\hat q}))
=\begin{bmatrix}\kappa(x) \sigma_+(e^{\ri \hat q})\kappa(x)^{-1}&\0_n\\
\0_n& \sigma_+(e^{\ri \hat q}) \end{bmatrix}K(\hat p,e^{\ri\hat q})
\begin{bmatrix} \sigma_+(e^{\ri \hat q}) &\0_n\\
\0_n&\ \sigma_+(e^{\ri \hat q}) \end{bmatrix}^{-1}
\label{H26}
\ee
is equivalent to the identities listed in Lemma \ref{lem:3.5}.
We see from this that $\hat K(z(\hat p,e^{\ri\hat q}))$ is a $G_\mu$-transform
of $K(\hat p, e^{\ri \hat q})$ \eqref{H3}, and thus $\hat K(z)$ belongs to
$\Phi_+^{-1}(\mu)$. Indeed, the right-hand side of \eqref{H26} can be written
as $\eta_L K(\hat p,e^{\ri\hat q}) \eta_R^{-1}$ with
\be
\eta_L=c\begin{bmatrix}\kappa(x)\sigma_+(e^{\ri \hat q})\kappa(x)^{-1}&\0_n\\
\0_n&\sigma_+(e^{\ri \hat q})\end{bmatrix},
\qquad
\eta_R=c\begin{bmatrix} \sigma_+(e^{\ri \hat q})&\0_n\\
\0_n&\sigma_+(e^{\ri \hat q})\end{bmatrix},
\label{H27}
\ee
where $c$ is a scalar ensuring $\det(\eta_L)=\det(\eta_R)=1$, and one can check
(see the last paragraph of Appendix A)
that this $(\eta_L,\eta_R)$ lies in the group $G_\mu$ \eqref{T13}.

To proceed further, we let $\hat K_o$ denote the restriction
of $\hat K$ to the dense open subset $\hat M_c^o$ and also let
$K_o\colon\cC_x\times\T_n\to\SL(2n,\C)'$ denote the map
defined by the corresponding restriction of the formula \eqref{H3}.
Notice that, in addition to \eqref{defred}, we have the relations
\be
\pi_\mu\circ\hat K_0=\pi_\mu\circ K_o \circ\cZ_x^{-1}
\quad\hbox{and}\quad
(\pi_\mu\circ K_o)^\ast\omega_M=\sum_{k=1}^nd\hat q_k\wedge d\hat p_k,
\label{H29}\ee
which follow from \eqref{H26} and the last sentence of Proposition \ref{prop:3.3}.
By using \eqref{densehatom} (together with
 $\hat K_o=\iota_\mu\circ\hat K_o$ and $K_o=\iota_\mu\circ K_o$)
the above relations imply the restriction of the equality
\be
(\pi_\mu\circ\hat K)^\ast\omega_M=\hat\omega_c
\label{H30}\ee
on $\hat M_c^o$. This equality is then valid on the full $\hat M_c$
since the 2-forms concerned are smooth.

It is a direct consequence of \eqref{H26} and
Proposition \ref{prop:3.2} that $\pi_\mu \circ \hat K$
is surjective.
Since, on account of \eqref{H30}, it is a local diffeomorphism,
it only remains to demonstrate
that the map $\pi_\mu \circ \hat K$ is injective. The relation
$\pi_\mu(\hat K(z))=\pi_\mu(\hat K(z'))$ for $z,z'\in\hat M_c$ requires that
\be
\hat K(z')=\begin{bmatrix}\eta_L(1)&\0_n\\\0_n&\eta_L(2)\end{bmatrix}
\hat K(z)\begin{bmatrix}\eta_R(1) &\0_n\\\0_n&\eta_R(2)\end{bmatrix}^{-1}
\label{H31}
\ee
for some $(\eta_L,\eta_R)\in G_\mu$. Supposing that \eqref{H31} holds,
 application of the decomposition $\hat K(z) = g_L(z) b_R(z)^{-1}$ to the formula \eqref{H25}
implies that
\be
\hat\alpha(z')=\eta_R(1)\hat\alpha(z)\eta_R(2)^{-1}
\label{H33}
\ee
and
\be
g_L(z')=\eta_Lg_L(z)\eta_R^{-1}.
\label{gleq}
\ee
The matrices on the two sides of \eqref{gleq} appear in the form \eqref{T17},
and standard uniqueness properties of the constituents in this generalized Cartan
decomposition now imply that
\be
\hat p(z')=\hat p(z)
\label{H32}
\ee
and
\be
\eta_R(1)=\eta_R(2)=m\in\T_n.
\label{H35}
\ee
We continue by looking at the $(k+1,k)$
components of the equality \eqref{H33} for $k=1,\dots,n-1$ using that $\hat\alpha_{k+1,k}$
depends on $z$ only through $\hat p(z)$ and it never vanishes.
(This follows from \eqref{H22}-\eqref{H23} by utilizing that
$\hat \theta(x,z)_{k,k+1} =\theta(x,\hat p(z))_{k,k+1}$ by \eqref{H20},
which is nonzero for each $\hat p(z)\in\bar\cC_x$ as seen from \eqref{A1}.)
Putting
\eqref{H35} into \eqref{H33}, we obtain that $m=C\1_n$ with a scalar $C$,
and therefore
\be
\hat\alpha(z')=\hat\alpha(z).
\label{H36}
\ee
The rest is an inspection of this matrix equality.
In view of \eqref{H32} and the forms of $\Delta(z)$ \eqref{H22} and $\hat \alpha(z)$ \eqref{H23},
the last column of the equality \eqref{H36} entails that
\be
\hat \theta(x,z)_{nk} =\hat \theta(x,z')_{nk},\quad k=2,\dots,n,
\ee
where we re-instated the dependence on $x$ that was suppressed above.
One can check directly from the formulas \eqref{H12}, \eqref{H20} and
\eqref{A1}, \eqref{A2} that
\be
\hat\theta(x,z)_{nk}=\bar z_{k-1}F_k(x,\hat p(z)),\quad k=2,\dots,n,
\ee
where $F_k(x,\hat p(z))$ is a smooth, strictly positive function.
Hence we obtain that $z_j=z_j'$ for $j=1,\dots,n-1$.
With this in hand,
since the variable $z_n$ appears only in $\Delta(z)$, we conclude from \eqref{H36}
that $\Delta(z)=\Delta(z')$. This plainly implies that $z_n=z'_n$, whereby
the proof is complete.

We note in passing that by continuing the above line of arguments the
free action of $G_\mu$ is easily confirmed. Indeed, for $z'=z$ \eqref{gleq}
also implies, besides \eqref{H35}, the equalities $\eta_L(2)=m$ and
$\eta_L(1)\kappa(x)\hat\zeta(x,z)^{-1}=\kappa(x)\hat\zeta(x,z)^{-1}m$.
Since $m=C\1_n$, as was already established, we must have
$(\eta_L,\eta_R)=C(\1_{2n},\1_{2n})\in\Z_{2n}$ \eqref{Z2n}.
By using that the image of $\hat K$ intersects every $G_\mu$-orbit, we can conclude that
$\bar G_\mu$ \eqref{barGmu} acts freely on $\Phi_+^{-1}(\mu)$.
\end{proof}

\begin{remark}
\label{rem:3.7}
Observe from Theorem \ref{thm:3.6} that
$\hat S=\{\hat K(z)\mid z\in\hat M_c\}$
 is a global cross-section for the action of $G_\mu$ on
$\Phi_+^{-1}(\mu)$. Hence $\hat S$ carrying the pull-back of $\omega$ as well as $(\hat M_c, \hat \omega_c)$
yield globally valid models of the reduced phase space $(M,\omega_M)$.
The submanifold of $\hat S$ corresponding to $\hat M_c^o$ \eqref{H13} is gauge equivalent
to $S^o$ \eqref{H7} that features in Proposition \ref{prop:3.3}.
\end{remark}

\section{Discussion}
\label{sec:4}

In this paper we clarified the global structure of the reduced phase space
$M$ \eqref{T16}, and thus completed the previous analysis \cite{M} that dealt
with the submanifold parametrized by
$\cC_x\times\T_n$. In terms of the model $\hat M_c$ \eqref{hatM} of $M$,
the complement of the submanifold in question is simply the zero set of the product of the complex variables.
The phase space $\hat M_c$ and the embedding of $\cC_x\times\T_n$ into it coincides
with what occurs for the so-called $\widetilde{\mathrm{III}}$-system of Ruijsenaars
\cite{RIMS95,FK1}, which is the action-angle dual of the standard trigonometric
Ruijsenaars-Schneider system. This circumstance is not surprising in light of the
fact \cite{M} that the reduced `main Hamiltonian' arising from $\cH_1$ \eqref{T4}
is a $\widetilde{\mathrm{III}}$-type Hamiltonian coupled to external fields.
We display this Hamiltonian below after exhibiting the corresponding Lax matrices.

The unreduced free Hamiltonians $\cH_j$, for any $j\in \Z^*$, mentioned in
Section \ref{sec:2}, can be written alternatively as
\be
\cH_j(K)=\frac{1}{2j}\tr(K\bJ K^\dag\bJ)^j=\frac{1}{2j}\tr(K^\dag\bJ K\bJ)^j.
\label{N1}
\ee
One can verify (for example by using the standard $r$-matrix formula of the Poisson
bracket on the Heisenberg double \cite{STS}) that the Hamiltonian flow generated by
$\cH_j$ reads
\begin{align}
K(t_j)&=\exp\bigg[\ri t_j\bigg(
(K(0)\bJ K(0)^\dag\bJ)^j
-\frac{1}{2n}\tr(K(0)\bJ K(0)^\dag\bJ)^j\1_{2n}\bigg)\bigg]K(0)\\
&=K(0)\exp\bigg[\ri t_j\bigg((\bJ K(0)^\dag\bJ K(0))^j
-\frac{1}{2n}\tr(\bJ K(0)^\dag\bJ K(0))^j\1_{2n}\bigg)\bigg].
\label{N2}
\end{align}
Since the exponentiated elements reside in the Lie algebra $\su(n,n)$,
these alternative formulas show that the flow stays in $\SL(2n,\C)'$,
as it must, and imply that the building blocks $g_L$ and $g_R$ of
$K=b_Lg_R^{-1}=g_Lb_R^{-1}$ follow geodesics on $\SU(n,n)$, while
$b_L$ and $b_R$ provide constants of motion.
Equivalently, the last statement means that
\be
K\bJ K^\dag\bJ=b_L\bJ b_L^\dag\bJ
\quad\text{and}\quad
K^\dag\bJ K\bJ=(b_R^{-1})^\dag\bJ b_R^{-1}\bJ
\label{N3}
\ee
stay constant along the unreduced free flows.

To elaborate the reduced Hamiltonians, note that for an element $K$ of the form
\eqref{T19} we have
\be
(b_R^{-1})^\dag\bJ b_R^{-1}\bJ
=\begin{bmatrix}
e^{-2v} \1_n &-e^{-v} \alpha\\
e^{-v}\alpha^\dag&e^{2v}\1_n-\alpha^\dag\alpha
\end{bmatrix}.
\label{N5}
\ee
By using this, as explained in Appendix \ref{sec:B}, one can prove that on
$\Phi_+^{-1}(\mu)$ the Hamiltonians $\cH_j$ can be written (for all $j$),
up to additive constants, as linear combinations of the expressions
\be
h_k=\tr(\alpha^\dag\alpha)^k,\qquad k=1,\dots,n.
\label{N6}
\ee
Since in this way the Hermitian matrix $L=\alpha^\dag\alpha$ generates the
commuting reduced Hamiltonians, it provides a Lax matrix for the reduced system.
By inserting $\alpha$ from \eqref{H5}, we obtain the explicit formula
\be
\begin{split}
L(\hat p,e^{\ri\hat q})=
&(e^{2v}+e^{-2u})e^{-2\hat p}+(e^{2v}+e^{-2v})\1_n\\
&-\sqrt{e^{-2u}e^{-2\hat p}+ e^{-2v}\1_n}e^{\ri\hat q}\theta(x,\hat p)^{-1}
e^v\sqrt{e^{-2\hat p}+\1_n}\\
&- e^v\sqrt{e^{-2\hat p}+\1_n}\theta(x,\hat p)e^{-\ri\hat q}
\sqrt{e^{-2u}e^{-2\hat p}+ e^{-2v}\1_n}.
\end{split}
\label{N7}
\ee
On the other hand, the Lax matrix of Ruijsenaars's
$\widetilde{\mathrm{III}}$-system can be taken to be \cite{RIMS95,FK1}
\be
\tilde L(\hat p,e^{\ri\hat q})
=e^{\ri\hat q}\theta(x,\hat p)^{-1}+\theta(x,\hat p)e^{-\ri\hat q}.
\label{N8}
\ee
The similarity of the
structures of these Lax matrices as well as the presence of the external field
couplings in \eqref{N7} is clear upon comparison. The extension of the Lax matrix
$\alpha^\dagger \alpha$ \eqref{N7} to the full phase space $M\simeq \hat M_c$ is of course given by
$\hat\alpha^\dag\hat\alpha$ by means of \eqref{H23}.

The main reduced Hamiltonian found in \cite{M} reads as follows:
\begin{multline}
 \cH_1(K(\hat p,e^{\ri\hat q}))=
-\frac{e^{-2u}+e^{2v}}{2}\sum_{j=1}^ne^{-2\hat p_j}+\\
\qquad+\sum_{j=1}^n\cos(\hat q_j)\big[1+ (1+e^{2(v-u)})e^{-2\hat p_j}
+e^{2(v-u)}e^{-4\hat p_j}\big]^{\tfrac{1}{2}}\prod_{\substack{k=1\\(k\neq j)}}^n
\bigg[1-\frac{\sinh^2\big(\frac{x}{2}\big)}{\sinh^2(\hat p_j-\hat p_k)}
\bigg]^{\tfrac{1}{2}}.
\label{N9}
\end{multline}
Liouville integrability holds since the functional independence of the involutive family
obtained by reducing $\cH_1,\dots,\cH_n$ \eqref{N1} is readily established
and the projections of the free flows \eqref{N2} to $M$ are automatically complete.
Similarly to its analogue in \cite{FG}, the Hamiltonian \eqref{N9} can be identified as an
Inozemtsev type limit of a specialization of van Diejen's 5-coupling deformation
of the hyperbolic $\BC_n$ Sutherland Hamiltonian \cite{vD}. This fact suggests that
it should be possible to extract the local form of dual Hamiltonians from \cite{vDE}
and references therein, which contain interesting results about closely related
quantum mechanical systems and their bispectral properties. Indeed, in several
examples, classical Hamiltonians enjoying action-angle duality correspond to
bispectral pairs of Hamiltonian operators after quantization. In a future work,
we wish to explore the action-angle dual of the Hamiltonian \eqref{N9} in
the reduction framework and employ the duality together with the traditional
projection method for studying the associated dynamics.

\bigskip
\begin{acknowledgements}
This work was supported in part by the Hungarian Scientific Research
Fund (OTKA) under the grant K-111697. The work was also partially supported by COST
(European Cooperation in Science and Technology) in COST Action MP1405 QSPACE.
\end{acknowledgements}

\appendix

\section{Explicit formulas for the matrices $\theta$, $\zeta$, and $\kappa$}
\label{sec:A}

In this appendix we collect the explicit expressions of the matrices $\theta$,
$\zeta$, and $\kappa$ for the reader's convenience. More detailed information
about these matrices can be found in the paper \cite{FK1}.
At an arbitrary point $\hat p\in\bar\cC_x$ \eqref{H1} the components of
$\theta(x,\hat p)$ are defined as follows
\be
\theta(x,\hat p)_{jk}=\frac{\sinh\big(\frac{x}{2}\big)}{\sinh(\hat p_k-\hat p_j)}
\prod_{\substack{m=1\\(m\neq j,k)}}^n\bigg[\frac{\sinh(\hat p_j-\hat p_m-\frac{x}{2})
\sinh(\hat p_k-\hat p_m+\frac{x}{2})}{\sinh(\hat p_j-\hat p_m)
\sinh(\hat p_k-\hat p_m)}\bigg]^{\tfrac{1}{2}},\quad j\neq k,
\label{A1}
\ee
and
\be
\theta(x,\hat p)_{jj}=\prod_{\substack{m=1\\(m\neq j)}}^n
\bigg[\frac{\sinh(\hat p_j-\hat p_m-\frac{x}{2})\sinh(\hat p_j-\hat p_m+
\frac{x}{2})}{\sinh^2(\hat p_j-\hat p_m)}\bigg]^{\tfrac{1}{2}}.
\label{A2}
\ee
Note that $\theta(x,\hat p)$ is an orthogonal matrix of determinant 1.
Next, with the help of the vector $r(x,\hat p)\in\R^n$ defined by
\be
r(x,\hat p)_j=\bigg[\frac{1-e^{-x}}{1-e^{-nx}}\bigg]^{\tfrac{1}{2}}
\prod_{\substack{k=1\\(k\neq j)}}^n
\bigg[\frac{1-e^{2\hat p_j-2\hat p_k-x}}{1-e^{2\hat p_j-2\hat p_k}}\bigg]^{\tfrac{1}{2}},
\quad j=1,\dots,n,
\label{A3}
\ee
the entries of the real $n\times n$ matrix $\zeta(x,\hat p)$ can be written as
\be
\begin{split}
&\zeta(x,\hat p)_{nn}=r(x,\hat p)_n,\quad
\zeta(x,\hat p)_{ij}=\delta_{ij}-\frac{r(x,\hat p)_ir(x,\hat p)_j}{1+r(x,\hat p)_n},\\
&\zeta(x,\hat p)_{in}=-\zeta(x,\hat p)_{ni}=r(x,\hat p)_i,\quad i,j\neq n.
\end{split}
\label{A4}
\ee
Finally, by introducing the vector $v=v(x)$:
\be
v(x)_j=\bigg[\frac{n(e^x-1)}{1-e^{-nx}}\bigg]^{\tfrac{1}{2}}e^{-\tfrac{jx}{2}},\quad j=1,\ldots,n,
\label{A5}
\ee
the elements of the $n\times n$ matrix $\kappa(x)$ read
\be
\begin{gathered}
\kappa(x)_{nn}=\frac{v(x)_n}{\sqrt{n}},\quad
\kappa(x)_{ij}=\delta_{ij}-\frac{v(x)_iv(x)_j}{n+\sqrt{n}v(x)_n},\\
\kappa(x)_{in}=-\kappa(x)_{ni}=\frac{v(x)_i}{\sqrt{n}},\quad i,j\neq n.
\end{gathered}
\label{A6}
\ee
 It can be shown that both
$\kappa(x)$ and $\zeta(x,\hat p)$ are orthogonal matrices of
determinant 1.
The main feature of $\kappa(x)$ is that (with $\nu(x)$ in \eqref{T12})
the matrix $\kappa(x)^{-1} \nu(x) \nu(x)^\dagger \kappa(x)$ is diagonal.
This implies that $\eta_L(1)=\kappa(x)\tau\kappa(x)^{-1}\in\UN(n)$
 satisfies \eqref{T15} for any $\tau\in \T_n$, which we used in the main text (see \eqref{H27}).
 In the above we assumed that $x>0$, otherwise the definition of the matrices $\zeta$ and $\kappa$ would
 need different formulas.

\section{On the reduced Hamiltonians}
\label{sec:B}

In this appendix we prove the claim, made in Section \ref{sec:4}, that on
the momentum surface $\Phi_+^{-1}(\mu)$ the Hamiltonians $\cH_j$, $j\in\Z^\ast$
\eqref{N1} are linear combinations of $h_k$, $k=1,\dots,n$ \eqref{N6}.
This will be achieved by establishing the form of the integer powers of the matrix
displayed in \eqref{N5}, which we denote here by $\cL$, i.e.
\be
\cL
=\begin{bmatrix}
e^{-2v}\1_n&-e^{-v}\alpha\\
e^{-v}\alpha^\dag&e^{2v}\1_n-\alpha^\dag\alpha
\end{bmatrix}.
 \label{B1}
\ee

\begin{lemma}
\label{lemma:B1}
For any positive integer $j$, the $j$-th power of the $2n\times 2n$ matrix $\cL$
\eqref{B1} reads
\be
\cL^j=\begin{bmatrix}
\cL^j_{11}&\cL^j_{12}\\
\cL^j_{21}&\cL^j_{22}
\end{bmatrix},
\label{B2}
\ee
where $\cL^j_{11},\cL^j_{12},\cL^j_{21},\cL^j_{22}$ are $n\times n$ blocks of
the form
\be
\begin{split}
&\cL^j_{11}=\sum_{m=1}^ja_m^{(j)}(\alpha\alpha^\dag)^{j-m},\quad
\cL^j_{12}=\alpha\sum_{m=1}^jb_m^{(j)}(\alpha^\dag\alpha)^{j-m},\\
&\cL^j_{21}=\alpha^\dag\sum_{m=1}^jc_m^{(j)}(\alpha\alpha^\dag)^{j-m},\quad
\cL^j_{22}=(-1)^j(\alpha^\dag\alpha)^j+\sum_{m=1}^jd_m^{(j)}(\alpha^\dag\alpha)^{j-m},
\end{split}
\label{B3}
\ee
with the $4j$ coefficients $a_m^{(j)},b_m^{(j)},c_m^{(j)},d_m^{(j)}$, $m=1,\dots,j$
depending only on the parameter $v$.
\end{lemma}

\begin{proof}
We proceed by induction on $j$. For $j=1$ the statement clearly holds,
and supposing that \eqref{B2}-\eqref{B3} is valid for some
fixed integer $j>0$ we simply calculate the $(j+1)$-th power
$\cL^{j+1}=\cL\cL^j$. This proves the statement.
\end{proof}

Our claim of linear expressibility follows at once, that is for any positive integer
$j$ we have
\be
\cH_j=(-1)^jh_j
+\sum_{k=1}^{j-1}\frac{k}{j}(a_{j-k}^{(j)}+d_{j-k}^{(j)})h_k
+\frac{n}{2j}(a_j^{(j)}+d_j^{(j)}).
\label{B9}
\ee
Incidentally, one also obtains a recursion for the
coefficients $a_m^{(j)},b_m^{(j)},c_m^{(j)},d_m^{(j)}$
from the proof of Lemma \ref{lemma:B1}.
If they are required,
this should enable one to establish the values of the constants that occur in \eqref{B9}.

As for the negative powers of $\cL$, one readily checks that the inverse of
$\cL$ is
\be
\cL^{-1}
=\begin{bmatrix}
e^{2v}\1_n-\alpha\alpha^\dag&e^{-v}\alpha\\
-e^{-v}\alpha^\dag&e^{-2v}\1_n
\end{bmatrix},
 \label{B10}
\ee
which has essentially the same form as $\cL$ does, thus the blocks of $\cL^{-j}$
$(j>0)$ can be expressed similarly as in Lemma \ref{lemma:B1}.
In fact, conjugating $\cL^{-1}$ with the $2n\times 2n$ involutory
block-matrix
\be
\bC=\begin{bmatrix}
\0_n&\1_n\\
\1_n&\0_n
\end{bmatrix},
\label{B11}
\ee
leads to the following formula
\be
\bC\cL^{-1}\bC
=\begin{bmatrix}
e^{-2v}\1_n&-e^{-v}\alpha^\dag\\
e^{-v}\alpha&e^{2v}\1_n-\alpha\alpha^\dag
\end{bmatrix},
\label{B12}
\ee
which implies that the blocks of $\cL^{-j}$ are obtained from those of $\cL^j$
by reversing their order and interchanging the role of $\alpha$ and $\alpha^\dag$.
Furthermore, since $\tr((\alpha\alpha^\dag)^k)=\tr((\alpha^\dag\alpha)^k)$ we get
\begin{equation}
\cH_{-j}=-\cH_j\qquad \forall j\in\Z^\ast.
\label{B13}
\end{equation}

\end{document}